\documentclass[a4paper,11pt]{article}

\usepackage{amsfonts,amssymb,amsmath,amsthm,dsfont,nicefrac,xspace}

\usepackage{fullpage,subfigure,float}
\usepackage{color}
\usepackage{graphicx}
\makeatletter
\let\NAT@parse\undefined
\makeatother
\usepackage[sort&compress, numbers]{natbib}
\usepackage{hyperref}

\usepackage[footnotesize]{caption}

\newcommand{\normsg}{\cl N_{{\rm sg},\alpha}}

\newcommand{\kappasg}{\kappa_{_{^{\rm sg}}}}
\newcommand{\Rbb}{\mathbb{R}}
\newcommand{\scp}[2]{\langle #1, #2 \rangle}

\newtheorem{proposition}{Proposition}
\newtheorem{corollary}{Corollary}

\newcommand{\supp}{{\rm supp}\,}
\newcommand{\tinv}[1]{{\textstyle\frac{1}{#1}}}
\newcommand{\sign}{{\rm sign}\,}

\newcommand{\ud}{\mathrm{d}}

\renewcommand{\leq}{\leqslant}
\renewcommand{\geq}{\geqslant}

\DeclareMathOperator{\rank}{rank}
\DeclareMathOperator{\ve}{vec}
\DeclareMathOperator{\tr}{tr}
\DeclareMathOperator{\conv}{conv}
\DeclareMathOperator{\st}{{s.\!t.}\xspace}

\DeclareMathOperator{\Id}{\mathds{1}}

\DeclareMathOperator*{\argmin}{argmin}

\newcommand{\bb}{\mathbb}
\newcommand{\range}[1]{{[\![#1]\!]}}

\newcommand{\bs}{\boldsymbol}
\newcommand{\cl}{\mathcal}
\newcommand{\ie}{\emph{i.e.}, }
\newcommand{\eg}{\emph{e.g.}, }

\newcommand{\bqmap}{{\bs A}}

\newcommand{\rv}{\mbox{r.v.}\xspace}
\DeclareRobustCommand{\amgis}{\text{\reflectbox{$\Sigma$}}}

\newcommand{\atomic}{\sharp}

\title{Consistent Basis Pursuit for Signal and Matrix Estimates\\
in Quantized Compressed Sensing}
\author{A. Moshtaghpour$^*$, L. Jacques$^*$, V. Cambareri$^*$, K. Degraux$^*$, C. De Vleeschouwer\thanks{ICTEAM
    institute, ELEN Department, Universit\'e catholique de Louvain
    (UCL), B1348 Louvain-la-Neuve, Belgium. VC,
    KD and CDV are funded by Belgian National Science Foundation
(F.R.S.-FNRS). AM is funded by the Walloon Region Mecatech project
SAVE. Copyright (c) 2015 IEEE.}}

\begin{document}
\maketitle

\begin{abstract}
This paper focuses on the estimation of low-complexity signals when they are observed
through $M$ uniformly
quantized compressive observations. Among such signals, we consider
1-D sparse vectors, low-rank matrices, or compressible signals
that are well approximated by
one of these two models. In this context, we prove the
estimation efficiency of a
variant of Basis Pursuit Denoise, called Consistent Basis Pursuit
(CoBP), enforcing consistency between the observations and the
re-observed estimate, while promoting its low-complexity nature. 
We show that the reconstruction error of CoBP decays like $M^{-1/4}$
when all parameters but $M$ are fixed. 
Our proof is connected to recent
bounds on the proximity of vectors or matrices when~\emph{(i)}~those
belong to a set of small intrinsic ``dimension'', as measured by the Gaussian mean
width, and \emph{(ii)} they share the same quantized (dithered) random
projections. By solving CoBP with
a proximal algorithm, we provide some extensive numerical observations
that confirm the
theoretical bound as $M$ is increased, displaying even faster
error decay than predicted. The same phenomenon is observed in the special,
yet important case of 1-bit CS. 
\end{abstract}

\section{Introduction}
\label{sec:introduction}

The theory of Compressed Sensing (CS) shows that many signals of interest can be
reconstructed from a few linear, and typically random,
observations~\cite{candes2006ssr,donoho2006cs,foucart2013mathematical}. Interestingly, this reconstruction is made possible if
the number of observations (or \emph{measurements}) is adjusted to the
intrinsic complexity of the signal, \eg its
sparsity for vectors or its low-rankness for matrices. Thus, this principle is a generalization of the
Shannon-Nyquist sampling theorem, where the sampling rate is set by the
bandwidth of the signal. However, a significant aspect of CS systems is the effect of
\emph{quantization} on the acquired observations, in particular for the purpose of 
compression and transmission~\cite{BoufChapter,GLPSY13,Kamilov12,Jacques2010,zymnis2010compressed,Dai2009,Dai2011}. 
This quantization is a non-linear transformation that both distorts the CS
observations and increases, especially at low bit rates, the reconstruction error of CS
reconstruction procedures.   

This work focuses on minimizing the impact of (scalar) quantization
during the reconstruction of a signal from its quantized compressive
observations. While more
efficient quantization procedures exist in the
literature (\eg $\Sigma\Delta$~\cite{GLPSY13},
universal~\cite{B_TIT_12}, binned~\cite{bib:Pai06,Kamilov12},
vector~\cite{bib:Pai06,vivekQuantFrame} or analysis-by-synthesis quantizations~\cite{shira2013}), scalar quantization remains
appealing for its implementation simplicity in most electronic devices,
and for its robustness against measurement lost.  

Conversely to other attempts, which consider quantization distortion as additive Gaussian
measurement noise \cite{candes2006near} and promote a Euclidean ($\ell_2$) fidelity with the
signal observations as in the Basis Pursuit Denoise (BPDN) program, better signal reconstruction methods are reached
by forcing \emph{consistency} between the re-observed signal estimate and
the quantized observations~\cite{Kamilov12,sun2009optimal,goyal_1998_lowerbound_qc}. 

We show here that a consistent version of the basis pursuit program \cite{Chen98atomic}, coined CoBP, provides
better signal estimates at large $M$ than those obtained by
BPDN. When reconstructing sparse or compressible signals,
CoBP is similar, up to an additional normalization constraint, to former
methods proposed in \cite{Dai2009,Dai2011,Jacques2010,Dirsksen-gap-RIP-sparse}. We prove the efficiency of CoBP from recent results on the proximity
of signals when those are taken in a set $\cl K \subset \bb R^N$ of small
``dimension'', \ie with small Gaussian
width $w(\cl K)$~\cite{ai2014one,chandrasekaran2012convex}, and when their
quantized random projections are consistent
\cite{jacques2014error,jacques2015small}. In particular, we show that for sub-Gaussian sensing
matrices, the $\ell_2$-reconstruction error of CoBP decays
as $\nicefrac{\sqrt{w(\cl K)}}{M^{1/4}}$, with an additional constant error bias
arising in the case of non-Gaussian sensing matrices. This contrasts with BPDN, whose 
reconstruction error is only guaranteed to saturate when $M$ increases.

The rest of this paper is structured as follows. Sec.~\ref{sec:problem-statement} introduces the problem
by explaining the low-complexity signal space, our Quantized
Compressed Sensing (QCS) model and the BPDN
reconstruction procedure as generally used in QCS. Sec.~\ref{sec:prox-conc-vect}
reviews important results on the proximity of consistent vectors; in
Sec.~\ref{sec:cons-basis-purs} we introduce and analyze CoBP. Finally,
Sec.~\ref{sec:experiments} demonstrates experimentally the
capabilities of this method in QCS of signals and matrices, before concluding.

\noindent\emph{Conventions:} Vectors and matrices are associated to bold
symbols. The probability of an event $\cl X$ is $\bb P(\cl
X)$. The identity matrix is $\Id_{D} \in \bb R^{D \times D}$ ($D
\in \bb N$), $\range{D}:=\{1,\,\cdots,D\}$ and $|\cl S|$ is the cardinality of $\cl S \subset \range{D}$. The $\ell_p$-norm of $\bs u$ is
$\|\bs u\|_p$ and the unit $\ell_p$-ball is $\bb B_p^{N}=\{\bs x\in\Rbb^N:
\|\bs x\|_p\leq 1\}$, with $\bb B^{N}:=\bb
B_2^N$. Assuming $N= n^2$ is a
square number, for a matrix $\bs U=(\bs u_1,\,\cdots,\bs u_n)  \in \bb R^{n \times
  n}$ with \emph{vectorization} $\ve(\bs U):=(\bs u^T_1,\,\cdots,\bs u^T_n)^T \in \bb R^N$, $\rank(\bs U)$, $\|\bs U\|$, $\|\bs U\|_*$ and $\|\bs U\|_F:=\tr(\bs U^T\bs
U)^{1/2}=\|\!\ve(\bs U)\|_2$ denote its
rank, operator norm, nuclear norm and its Frobenius norm,
respectively.  We will often assimilate
matrices in $\bb R^{n \times n}$ with their vectorization in $\bb R^{N}$, \eg identifying $\{\bs U \in \bb R^{n \times n}:
  \|\bs U\|_F \leq 1\}$ with $\bb B^N$. Finally, we write $f\lesssim
g$ or $f = O(g)$ if $f \leq c\,g$ for~$c>0$, and similarly for $f\gtrsim g$ and $f = \Omega(g)$.

\section{Quantized Compressed Sensing of Low-Complexity Signals}
\label{sec:problem-statement}

\subsection{Low-complexity Signal Model}
\label{sec:low-compl-sign}
 
This work focuses on the sensing of signals belonging to a low-complexity
set $\cl K \subset \bb R^N$. A typical example is the set of $K$-sparse vectors $\cl K = \Sigma_K:=\{\bs u \in \bb R^N: \|\bs
u\|_0 := |\supp \bs u| \leq K\}$, as well as the set of rank-$r$ matrices
$\cl C_r := \{\bs U \in \bb R^{n \times n} \simeq \bb R^N:
\rank(\bs U) \leq r\}$. 

As in \cite{chandrasekaran2012convex}, we assume that the (bounded) \emph{convex hull} $\overline {\cl K} := \conv({\cl K
\cap \bb B^N})$ of $\cl K$ is associated to the definition of an
appropriate \emph{atomic} norm\footnote{If
  $\cl K$ is convex and centrally symmetric around the origin,
  ${\|\!\cdot\!\|_\atomic}$ can always be defined by the \emph{gauge} of
  $\cl K$ (see \cite{chandrasekaran2012convex} for details).}
${\|\!\cdot\!\|_\atomic}$ such that 
\begin{equation}
  \label{eq:K-s-def}
\overline {\cl K} = \overline {\cl K}_s := \{\bs u \in \bb R^N: \|\bs u\|_\atomic \leq s, \|\bs
u\|_2 \leq 1\},  
\end{equation}
for some $s>0$. For instance, for compressible signals in $\overline{\Sigma}_K$, ${\|\!\cdot\!\|_\atomic} =
{\|\!\cdot\!\|_1}$ and $s=\sqrt K$, while for matrices in $\overline
{\cl C}_r$,
${\|\!\cdot\!\|_\atomic} = {\|\!\cdot\!\|_\ast}$ for~$s = \sqrt r$~\cite{plan2012robust}.

The ``low-complexity'' nature of these sets
stems from their small \emph{Gaussian mean width} 
$$
w(\cl K) := \bb E\sup_{\bs u \in \cl K} |\bs g^T \bs u|,\quad \bs g \sim
\cl N(0, \Id_N).
$$
For instance, $w(\Sigma_K)^2 = w(\overline \Sigma_K)^2 \lesssim K \log
\nicefrac{N}{K}$ and $w(\cl
C_r)^2 = w(\overline{\cl C_r})^2 \leq 4 n r$~\cite{ai2014one,chandrasekaran2012convex}.
The quantity $w(\cl K)$, also called Gaussian complexity,
has been recognized as central, \eg for random
processes characterization \cite{vaart1996weak}, high-dimensional
statistics and inverse problem solving \cite{chandrasekaran2012computational,chandrasekaran2012convex} or
classification in randomly projected domains~\cite{bandeira2014compressive}. As explained below, $w(\cl K)$ also determines
the minimal number of measurements for CS of
signals in $\cl K$~\cite{chandrasekaran2012convex}.

\subsection{Quantized Compressed Sensing}
\label{sec:quant-compr-sens}

 Given a certain quantization resolution $\delta>0$, we
focus on the impact of a uniform (midrise) quantizer
$\cl Q(t) := \delta (\lfloor\tfrac{t}{\delta}\rfloor + \frac{1}{2})
 \in \bb Z_\delta := \delta (\bb Z + \frac{1}{2})$, applied
 componentwise, in the quantized sensing model
\begin{equation}
  \label{eq:psi-def}
  \bs q = \bqmap(\bs x_0) := \cl Q(\bs\Phi \bs x_0 + \bs \xi) \in \bb Z^M_\delta,  
\end{equation}
where $\bs \Phi \in \bb R^{M\times N}$ is a random \emph{sensing matrix}
and $\bs \xi \sim \cl U^M([-\delta/2, \delta/2])$ (\ie $\xi_i
\sim_{\rm iid} \cl U([-\delta/2, \delta/2])$ for $i \in \range{M}$) is a uniform
dithering\footnote{As in \cite{jacques2015small}, our results remain
  valid if $\bs \xi \sim \cl U^M([t, t+\delta])$ for any $t\in\bb
  R$.}. This random dithering is known at the signal reconstruction and stabilizes 
the action of $\cl Q$~\cite{B_TIT_12,Gray98,jacques2013quantized}. By slightly abusing the
notation, when \eqref{eq:psi-def} senses an element $\bs X_0$ of a matrix set
in $\bb R^{n\times n}$, $\bs
x_0 = \ve(\bs X_0)$ amounts to the $N$-length vectorization of this element, assuming
$N=n^2$.  

As often the case in CS, we consider that $\bs \Phi$ is
a sub-Gaussian random matrix, \ie its entries are distributed as $
\Phi_{ij} \sim_{\rm iid} \varphi$ with $\varphi$ a symmetric,
zero-mean and unit-variance sub-Gaussian random variable (\rv), having
finite sub-Gaussian norm 
$$
\textstyle \|\varphi\|_{\psi_2} := \sup_{p\geq 1} p^{-1/2}(\bb
E|\varphi|^p)^{1/p} < \infty.
$$
For such a \rv of sub-Gaussian norm $\alpha>0$, we have in fact $\bb
P[|\varphi| > t] \lesssim \exp(-c t^2/\alpha^2)$ for any $t>0$.
Examples of such \rv's are Gaussian, uniform, bounded or Bernoulli
distributed \rv's. Below, we write $\varphi \sim \normsg(0,1)$, and
the shorthand $\bs \Phi \sim \normsg^{M\times N}(0,1)$ for the
associated $M\times N$ matrix, to specify that
$\varphi$ is a sub-Gaussian \rv of norm $\alpha$. 

In the absence of quantization, if $M \gtrsim w(\cl
K)^2$, with high probability, any $\bs x_0 \in \cl K$ can be reconstructed from
sub-Gaussian observations $\bs \Phi \bs x_0$ using convex optimization programs
such as Basis Pursuit~\cite{chandrasekaran2012convex}. Therefore, the
minimal number of measurements needed for reconstructing $K$-sparse
or compressible signals in $\bb R^N$ grows like $K \log
\nicefrac{N}{K}$, and like $n r$ for rank-$r$ and compressible $n\times n$ matrices~\cite{ai2014one,chandrasekaran2012convex}.

\paragraph{1-bit Quantization Regime:} The
exponentially decaying tail bounds of the sub-Gaussian entries of $\bs
\Phi$ show that a suitable
value of $\delta$ can essentially turn \eqref{eq:psi-def} into
a 1-bit CS model when $\cl K$ is
bounded \cite{jacques2013robust,plan2013one,BB_CISS08}. Indeed, from
the definition of $\cl Q$ and assuming $\|\bs x_0\|_2=1$, for $i\in
\range{M}$, $\bb P[q_i
\notin \{\pm \delta/2\}] = \bb P[|\bs\varphi_i^T \bs x + \xi_i| \geq
\delta ] = p_0 \leq 2\exp(-\tinv{2}\delta^2)$, with $p_0=0.0027$ for
$\delta = 3$. Our study holds in such a regime with
the interesting advantage of allowing the estimation of the signal
norm, as opposed to the 1-bit CS model $\sign(\bs\Phi\bs
x_0)$~\cite{BB_CISS08,plan2012robust}. This is due to the
pre-quantization dithering in \eqref{eq:psi-def}. Interestingly, combining the sign operator with prequantization thresholds
in 1-bit CS also removes this signal norm uncertainty~\cite{knudson2014one}.

\subsection{Basis Pursuit Denoise}
\label{sec:basis-purs-deno}

 The first method used to estimate
$\bs x_0$ from $\bs q$ in \eqref{eq:psi-def} was considering quantization as an additive noise of bounded power under high
resolution assumption (HRA), \ie $\delta \ll \|\bs
x_0\|_2$~\cite{candes2006ssr}. In \eqref{eq:psi-def}, the impact of the dithering
provides $\bs q = (\bs \Phi \bs x
+ \bs \xi) + \bs n$ with $\bs n := \bs q -  (\bs \Phi \bs x
+ \bs \xi) \sim \cl U^M([-\delta/2,\delta/2])$. Therefore, $\|\bs n\|_2^2 \leq \epsilon^2 =
M\frac{\delta^2}{12} + \kappa\sqrt M$ holds with high probability for
small $\kappa$ (\eg $\kappa=2$)~\cite{Jacques2010,Gray98}. In such a case, the general BPDN program, 
\begin{equation*}
 \bs x^*_{\rm BPDN} := \argmin_{\bs u \in \bb R^N} \|\bs u\|_\atomic\ \st\
  \|\bs\Phi \bs u + \bs \xi - \bs q\|_2 \leq \epsilon  
  \eqno{\text{({\small  BPDN})}}
\end{equation*}
can be solved for estimating $\bs x_0$. When $\nicefrac{\bs
  \Phi}{\sqrt M}$
satisfies the restricted isometry property (RIP) and when $\cl K$ is the set
of sparse signals, then, setting
${\|\!\cdot\!\|_\atomic}={\|\!\cdot\!\|_1}$,~\cite{GLPSY13,Jacques2010} show that 
\begin{equation*}
  \|\bs x^* - \bs x_0\|_2 = O(\nicefrac{\epsilon}{\sqrt M}) = O(\delta).
\end{equation*}
A similar result holds in the case of QCS of low-rank matrices using a
Lasso reconstruction that minimizes a Lagrangian formulation of
BPDN~\cite{candes2011tight}. Notice that a variant of BPDN,
  called Basis Pursuit DeQuantizer of moment $p$ (BPDQ \cite{Jacques2010}), replaces the
  $\ell_2$-norm of the BPDN constraint by an
  $\ell_p$-norm ($2\leq p < \infty$). Its error decays like 
  $O(\delta/\sqrt{\log M})$~\cite{BoufChapter}.

\section{Proximity of Consistent Vectors}
\label{sec:prox-conc-vect}

This section summarizes a recent study showing that the proximity of
vectors of a subset $\cl K \subset \bb R^N$ with small Gaussian mean
width can be bounded provided they share the same image through the random mapping
$\bqmap$, \ie if they are consistent~\citep{jacques2015small}. As will be clear in Sec.~\ref{sec:cons-basis-purs}, this
property is the key for characterizing the behavior of CoBP.  

This proximity is impacted by the level of \emph{anisotropy} of the 
sub-Gaussian rows composing $\bs \Phi \sim \normsg^{M\times
  N}(0,1)$~\citep{ai2014one}, as measured by the smallest $\kappasg
\geq 0$ such that, for $\bs \varphi \sim \normsg^N(0,1)$, $\bs g \sim
\cl N^{N}(0,1)$ and all $\bs u \in \bb R^N$,
\begin{equation}
  \label{eq:Berry-Esseen-relation}
\!\!\textstyle\int_{0}^{+\infty}\big|\bb P(|\scp{\bs \varphi}{\bs u}| \geq t) -
  \bb P(|\scp{\bs g}{\bs u}| \geq t)\big|\,\ud t\,\leq\,\kappasg \|\bs u\|_{\infty}.
\end{equation}
For Gaussian (isotropic) random vectors $\kappasg = 0$, while for
sub-Gaussian $\bs \varphi \sim \normsg^N(0,1)$, $\kappasg \leq 9\sqrt{27}\,\alpha^3$, with $\alpha \leq 1$ for Bernoulli \rv's~\citep{jacques2015small}. 

As clarified in Prop.~\ref{prop:consistency-width}, when the mapping
$\bqmap$ integrates a non-Gaussian, but sub-Gaussian sensing
matrix~$\bs \Phi$, the proximity of consistent elements
$\bs x,\bs y$ in $\cl K$ is guaranteed when $\bs x - \bs y$ is not
``too sparse'', \ie when it belongs to 
\begin{equation*}
\amgis_{K_0} := \{\bs u \in \bb R^N: K_0\|\bs u\|^2_\infty \leq \|\bs
u\|_2^2\},  
\end{equation*}
for $K_0$ large enough compared to $\kappasg^2$. For instance, a $K$-sparse
vector $\bs u \in \Sigma_K := \{\bs v:
\|\bs v\|_0 := |\supp \bs v| \leq K\}$ cannot belong to $\amgis_{K_0}$
for $K_0 > K$ as then $\|\bs
u\|_2^2 \leq K \|\bs u\|^2_\infty$. 

\begin{proposition}[Consistency width \cite{jacques2015small}]
\label{prop:consistency-width}
Given a quantization resolution $\delta >0$, $\epsilon \in (0,1)$, a sub-Gaussian distribution $\normsg(0,1)$ respecting
\eqref{eq:Berry-Esseen-relation} for $0 \leq \kappasg < \infty$,
and $\cl K \subset \bb B^{N}$ a
bounded subset of $\bb R^N$, there exist some values
$C,c>0$ depending only on $\alpha$ and such that, if 
\begin{equation}
  \label{eq:prop-consist-width-minimal-cond}
  M \geq C\,\tfrac{(2 + \delta)^4}{\delta^2 \epsilon^4} \, w(\cl K)^2,    
\end{equation}
then, for $\bs \Phi \sim \normsg^{M\times N}(0,1)$, $\bs \xi \sim
\cl U^M([-\delta/2, \delta/2])$ and $\sqrt K_0 \geq 16 \kappasg$, with probability exceeding $1 - 2\exp(-c \epsilon M/(1+\delta))$, we have for all $\bs x,\bs y \in
\cl K$
\begin{equation}
  \label{eq:consistency-width}
  \bs x - \bs y \in \amgis_{K_0},\ \bqmap(\bs x) = \bqmap(\bs y)\quad \Rightarrow\quad \|\bs x - \bs y\|_2 \leq \epsilon, 
\end{equation}
with $\bqmap$ defined in \eqref{eq:psi-def}. Moreover, for any orthonormal basis
$\bs \Psi \in \bb R^{N \times N}$, if $\cl K = (\bs \Psi \Sigma_K)
\cap \bb B^N$ then \eqref{eq:prop-consist-width-minimal-cond} simplifies to
\begin{equation}
  \label{eq:prop-consist-width-minimal-cond-K-sparse}
  M \geq C'\,\tfrac{2 + \delta}{\epsilon}\, K
  \log\big(\tfrac{N}{K \delta}(\tfrac{2 +
    \delta}{\epsilon})^{3/2}\big),
\end{equation}
for some $C'>0$ depending only on $\alpha$.
\end{proposition}
We remark that for Gaussian sensing matrices, the
``antisparse'' condition on $\bs x-\bs y$ (and on $K_0$) vanishes
since $\kappasg = 0$. This provides, in the special case of the sparse signal
  set, a proximity bound in \eqref{eq:consistency-width} formerly established in \cite{jacques2014error}.  
 
\section{Consistent Basis Pursuit}
\label{sec:cons-basis-purs}

The previous sections allow us now to define a
suitable reconstruction
procedure for estimating any signal $\bs x_0 \in \overline {\cl K}_s$
(for some $s>0$ in~\eqref{eq:K-s-def}) observed through the
model~\eqref{eq:psi-def}, \eg for reconstructing compressible
signals or matrices belonging to $\overline{\Sigma}_K$ or
$\overline{\cl C}_r$, respectively. We split the study according to the nature
of the sensing matrix. 

\subsection{Gaussian Sensing Matrix} 
\label{sec:gauss-sens-matr}

When $\bs \Phi$ is Gaussian, \ie $\kappasg = 0$, we propose to
estimate $\bs x_0$ with the
following program coined Consistent Basis Pursuit,
\begin{equation*}
  \bs x^* := \argmin_{\bs u \in \bb R^N} \|\bs u\|_\atomic\ \st\ 
    \bqmap(\bs u) = \bqmap(\bs x_0),\ \bs u \in \bb B^N.
  \eqno{\text{({\small  CoBP})}}
\end{equation*}
This is a convex optimization as the first constraint is
  equivalent to $\|\bs\Phi \bs u + \bs \xi - \bqmap(\bs x_0)\|_\infty
  \leq \delta/2$~\cite{Jacques2010}. The proximity of $\bs x^*$ to
$\bs x_0$ is then guaranteed by
Prop.~\ref{prop:consistency-width}. 
\begin{proposition}
\label{prop:cons-basis-purs-stab-Gauss}
If $\bqmap$ respects \eqref{eq:consistency-width} for all $\bs
x,\bs y \in \overline{\cl K}_s$ and $K_0 = 0$, then for all $\bs x_0 \in \overline{\cl K}$, the
estimate $\bs x^*$ obtained by CoBP from $\bs q = \bqmap(\bs x_0)$ satisfies
$
\|\bs x_0 - \bs x^*\|_2 \leq \epsilon
$. 
\end{proposition}
\begin{proof}
Since $\bs x_0 \in \overline{\cl K}_s$ is a feasible vector of the CoBP constraints, we necessarily
have $\|\bs x^*\|_\atomic \leq \|\bs x_0\|_\atomic \leq s$. By definition
of CoBP,
$\bs x^* \in \bb B^N$ so that $\bs x^* \in \overline{\cl K}_s$. The result follows from~\eqref{eq:consistency-width} with $\bs x
= \bs x_0$ and $\bs y = \bs x^*$.  
\end{proof}

Prop.~\ref{prop:cons-basis-purs-stab-Gauss} assumes that $K_0=0$ in \eqref{eq:consistency-width}. This
holds if $\kappasg = 0$, \eg if $\bs \Phi \sim \cl N^{M \times N}(0,1)$. Therefore, combining
the conditions of Prop.~\ref{prop:consistency-width} with this last
proposition, we get the following corollary by saturating~\eqref{eq:prop-consist-width-minimal-cond} with respect to
$M$.
\begin{corollary}
\label{cor:decay-gauss-error}
Given some universal constant $c>0$, with probability exceeding
$1-2\exp(-c M^{3/4}/\sqrt \delta)$ over the draw of $\bs \Phi \sim
\cl N^{M\times N}(0,1)$ and
$\bs \xi \sim \cl U^M([-\delta/2,\delta/2])$, for every $\bs x_0 \in
\overline{\cl K}_s$, the estimate $\bs x^*$ obtained by CoBP from $\bs
q = \bqmap(\bs x_0)$ satisfies 
$$
\textstyle \|\bs x_0 - \bs x^*\|_2 = O\big( \frac{2+\delta}{\sqrt \delta}\,(\tfrac{w(\overline{\cl K}_s)^2}{M})^{1/4} \big),
$$  
\ie $\|\bs x_0 - \bs x^*\|_2 = O\big(M^{-1/4}\big)$ if only $M$ varies.
\end{corollary}
At first sight, the error decay of CoBP in $M^{-1/4}$ could seem
slow. However, as mentioned in Sec.~\ref{sec:basis-purs-deno}, the best known error decay for
BPDN under the sensing model~\eqref{eq:psi-def}
is~$O(\delta)$~\cite{GLPSY13}, which does not decay with $M$. The same constant bound was
  found for a variant of  CoBP without the ball constraint~\cite{Dirsksen-gap-RIP-sparse}.

\begin{figure*}[t!]
  \centering
\newlength{\figwdth}
\setlength{\figwdth}{.45\textwidth}
\subfigure[Gaussian QCS of sparse
    signals.]{\includegraphics[width=\figwdth]{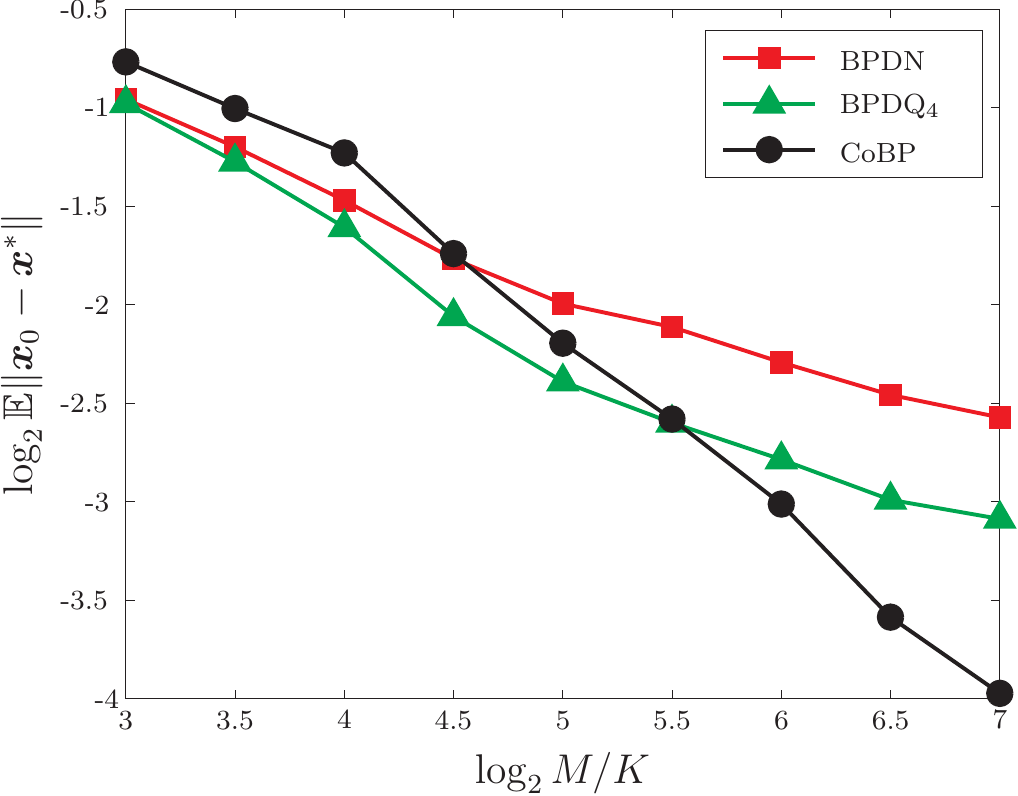}}
\subfigure[Bernoulli \emph{vs} Gaussian QCS of sparse
    signals.]{\raisebox{.5mm}{\includegraphics[width=\figwdth]{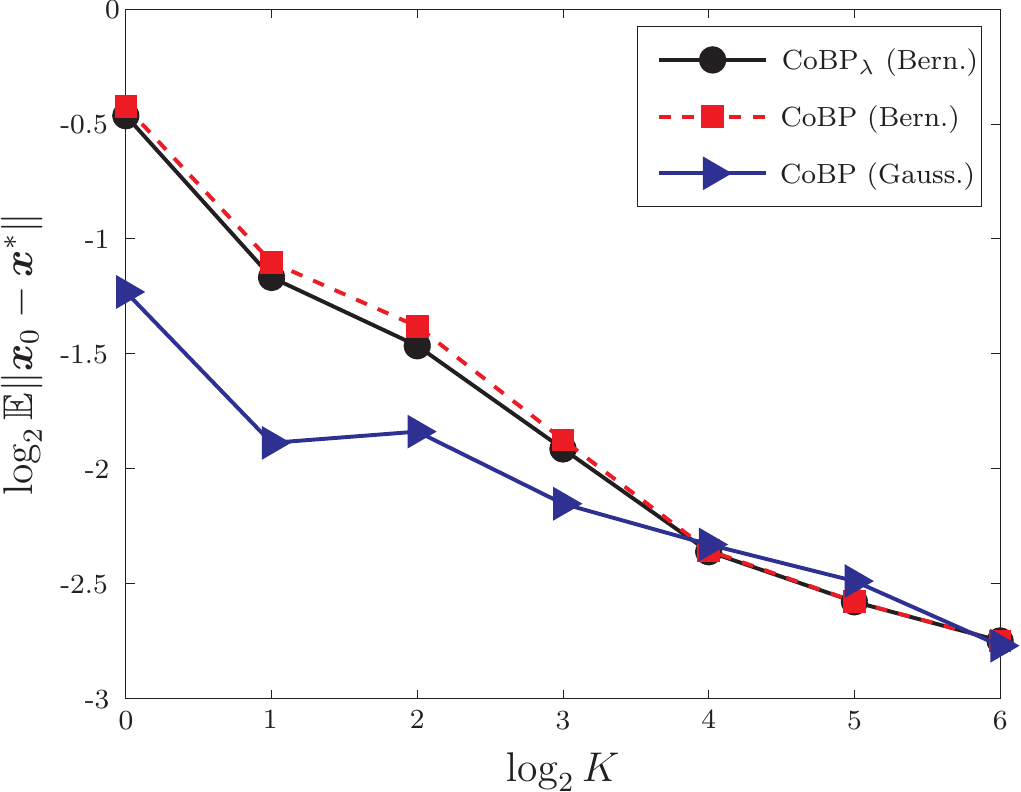}}}\\[4mm]
\subfigure[Gaussian QCS of rank-1
matrices.]{\raisebox{.25mm}{\includegraphics[width=\figwdth]{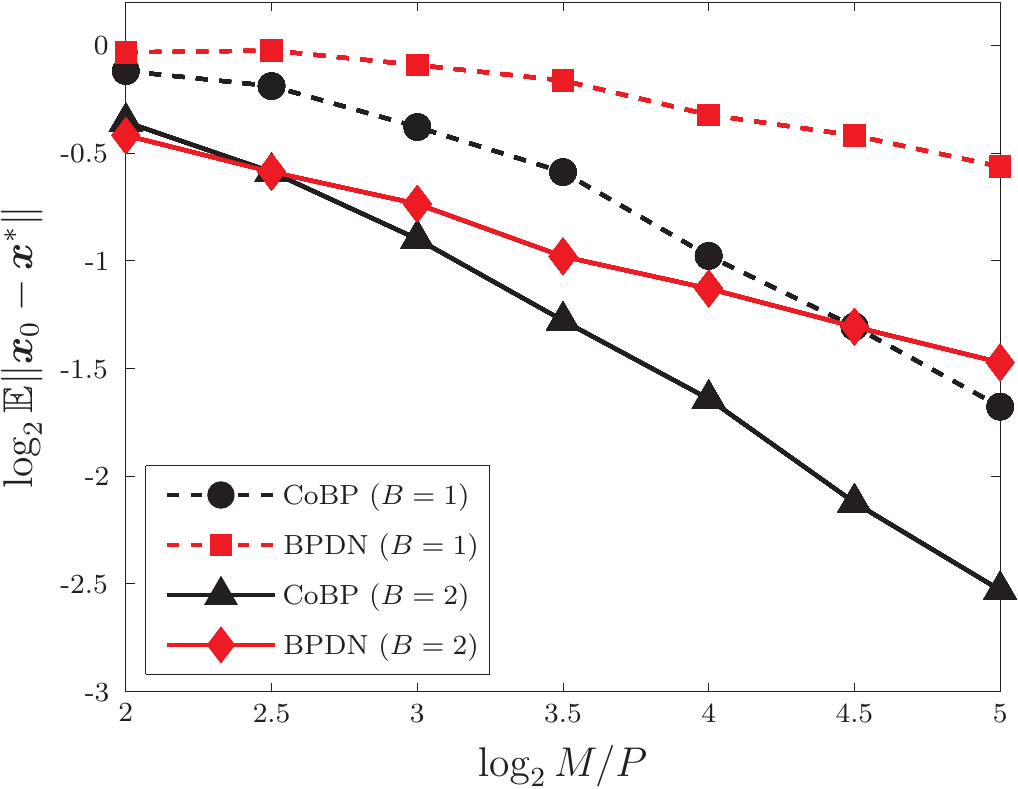}}}
\label{QCS-Gauss-sparse}
\end{figure*}

\subsection{Non-Gaussian Sensing Matrix} 

For non-Gaussian $\bs \Phi$, $\kappasg \neq 0$
in general. In order to reach a meaningful estimate of $\bs x_0 \in
\overline{\cl K}_s$, we further assume that $\|\bs x_0\|_\infty \leq \lambda$, for some
$\lambda>0$. As will be clear, this allows us to characterize the sparse nature of $\bs
x_0 - \bs x^*$ when $\bs x^*$ is an estimate of $\bs x_0$ produced by
the modified program:
\begin{equation*}
  \bs x^* := \argmin_{\bs u \in \bb R^N} \|\bs u\|_\atomic\ \st\
  \begin{cases}
    \bqmap(\bs u) = \bqmap(\bs x_0),\\
    \bs u \in \bb B^N \cap \lambda \bb B^N_\infty,
  \end{cases}
\eqno{({\rm CoBP}_{\lambda})}
\end{equation*}
with ${\bf CoBP}_{\lambda} \equiv {\bf CoBP}$ as soon as $\lambda > 1$
since $\bb B^N \subset \bb B^N_\infty$.

\begin{proposition}
\label{prop:cons-basis-purs-stab}
If $\bqmap$ respects \eqref{eq:consistency-width} for all $\bs
x,\bs y \in \overline{\cl K}_s$ and any $K_0 \geq (16 \kappasg)^2$,
then for any $\bs x_0 \in \overline{\cl K}_s \cap \lambda \bb
B_\infty^N$, the solution obtained by CoBP$_{\lambda}$ from $\bs q =
\bqmap(\bs x_0)$ respects
$$
\|\bs x_0 - \bs x^*\|_2 \leq \epsilon + 2\lambda \sqrt{K_0}.
$$ 
\end{proposition}
\begin{proof}
As for the proof of Prop.~\ref{prop:cons-basis-purs-stab-Gauss},  $\bs
x_0\in \overline{\cl K}_s$ implies that $\bs x^* \in
\overline{\cl K}_s$. If $\|\bs x_0 - \bs x^*\|_2 \leq \sqrt{K_0} \|\bs x_0 - \bs
x^*\|_\infty$, then, since $\bs x_0,\bs x^* \in \lambda \bb
B^N_\infty$, $\|\bs x_0 - \bs x^*\|_2 \leq 2\lambda
\sqrt{K_0}$. Otherwise, we have $\bs x_0 - \bs x^* \in
\amgis_{K_0}$. In this case, since \eqref{eq:consistency-width} is
assumed satisfied for all pairs of vectors
of $\overline{\cl K}_s$, we have 
$\|\bs x_0 - \bs x^*\|_2 \leq \epsilon$, which concludes the proof.  
\end{proof}

\noindent Taking $K_0 = \lceil (16 \kappasg)^2 \rceil$, this corollary is easily established.
\begin{corollary}
\label{cor:decay-nongauss-error}
Given some universal constant $c>0$, with probability exceeding
$1-2\exp(-c \frac{M^{3/4}}{\sqrt \delta})$ over the draw of $\bs \Phi \sim
\normsg^{M\times N}(0,1)$ and
$\bs \xi \sim \cl U^M([-\delta/2,\delta/2])$, the CoBP$_{\lambda}$ estimate $\bs
x^*$ of any $\bs x_0
\in \overline{\cl K}_s\cap\lambda \bb B^N_\infty$ satisfies 
\begin{equation}
  \label{eq:non-gaussian-rec-error}
  \textstyle \|\bs x_0 - \bs x^*\|_2 = O\big(\frac{2+\delta}{\sqrt
    \delta}\,(\tfrac{w(\overline{\cl K}_s)^2}{M})^{1/4} + \kappasg \lambda\big),  
\end{equation}
\ie $\|\bs x_0 - \bs x^*\|_2 = O(M^{-1/4} + \kappasg \lambda)$ if only
$M$ varies.  
\end{corollary}
Loosely speaking, Cor.~\ref{cor:decay-nongauss-error} shows that the reconstruction error
is not guaranteed to decay below a certain level fixed by
$\kappasg\|\bs x_0\|_\infty$. A similar behavior was already observed
in the case of 1-bit CS with non-Gaussian measurements~ \cite{ai2014one}.

\section{Experiments}
\label{sec:experiments}

In this section we run several numerical simulations in order to assess the
experimental benefit of CoBP compared to BPDN in various QCS
settings. As CoBP is a convex optimization problem containing non-smooth convex
functions, we solve\footnote{Free matlab code: \url{http://sites.uclouvain.be/ispgroup/index.php/Softwares}.} it with the versatile Parallel Proximal Algorithm (PPXA)
\cite{combettes2011proximal}, this one being efficiently implemented
in the UNLocBoX toolbox~\cite{perraudin2014unlocbox}. We refer the
reader to \cite{golbabaee2012hyperspectral} for an example application
of PPXA in the solution of low-rank matrix recovery. 

For our experiments, three different sensing contexts are tested: the two first ones consider QCS of sparse signals (for
Gaussian or Bernoulli sensing matrices), while
the last one focuses on QCS of rank-1 matrices.  In all cases, the
quantization resolution is fixed by $\delta =
6\times
2^{1-B}$ with $B\in \range{4}$. As explained in
Sec.~\ref{sec:quant-compr-sens}, each $q_i$ can then be
essentially coded with $B$ bits, \eg if $B=1$, $\bb E|\big\{i:
q_i \notin \{\pm \delta/2\}\big\}| \leq 0.0027 M$. Some of our results are
compared to those of BPDN with $\epsilon$ set as in
Sec.~\ref{sec:basis-purs-deno}. The constraint ``$\bs u \in \bb B^N$''
is also added to BPDN for reaching fair comparisons with
CoBP\footnote{The ratio of computational times between CoBP and BPDN is
about~$1.3$.}.

\subsection{Gaussian QCS of sparse signals} 

In this
experiment, we set $N=2048$, $K=16$, $B=3$ and $M/K \in [8,128]$, \ie well after the \emph{phase transition} (here around $M/K
\simeq 6$) where sparse signal reconstruction from noisy CS measurements is guaranteed~\cite{candes2006near}. 
For each value of $M$, 20 different Gaussian sensing matrices,
dithering realizations and unit-norm $K$-sparse signals were randomly
generated. Each signal $\bs x_0$ has its $K$-length support selected uniformly at random in $\range{N}$,
with non-zero components drawn as $\cl N(0,1)$ before
normalization. The reconstruction error decay
averaged over these 20 trials is
shown for BPDN, BPDQ with $p=4$ (see Sec.~\ref{sec:basis-purs-deno}) and CoBP in Fig.~\ref{QCS-Gauss-sparse}(left) in
a $\log_2/\log_2$ plot. For indication, a linear fitting over the last
4 values of $\log_2 M/K$ provides slopes of value $-0.31$, $-0.33$  and
$-0.95$ for BPDN, BPDQ and CoBP, respectively. As already observed
 experimentally in other works forcing tight or approximate
 consistency in signal reconstruction
 \cite{Kamilov12,goyal_1998_lowerbound_qc,Dai2009,Dai2011,Jacques2010}, this clearly highlights
the advantage of consistent signal reconstruction when $M/K$ is large. Moreover, CoBP approaches an
error decay of $M^{-1}$ similar to the distance decay of consistent
$K$-sparse vectors when
\eqref{eq:prop-consist-width-minimal-cond-K-sparse} is saturated, \ie better than the
``$M^{-1/4}$'' of Cor.~\ref{cor:decay-gauss-error}.

\subsection{Bernoulli \emph{vs} Gaussian QCS}
\label{sec:bern-emphvs-gauss}

This second experiment stresses the impact of the sub-Gaussian
nature of the sensing matrix over the CoBP reconstruction error. We focus
on the case of Gaussian QCS ($\bs \Phi \sim \cl N^{M \times
  N}(0,1)$) and Bernoulli QCS (\ie
$\Phi_{ij}$ equals $\pm 1$ with probability $1/2$) when observing
$K$-sparse signals for $K$ growing and $M/K$ constant. In particular, we set
$N=1024$, $B=4$, $K\in [1,64]$ and $M/K=16$. For
each value of $K$, 20
different sensing matrices, dithering realizations and unit-norm $K$-sparse
signals are generated as in the first experiment. CoBP and CoBP$_\lambda$ are compared
(with an oracle assisted $\lambda := \|\bs x_0\|_\infty$). Comparing
the error bounds for Gaussian and sub-Gaussian QCS
in Cor.~\ref{cor:decay-gauss-error} and in
Cor.~\ref{cor:decay-nongauss-error}, respectively, we expect that at
low $K$ and for $M/K$ constant, Bernoulli QCS reaches worst reconstruction
error than Gaussian QCS, as then the bias $\kappasg \lambda
= \kappasg \|\bs x_0\|_\infty \simeq
\kappasg /\sqrt{K}$ can be high. This is indeed observed in
Fig.~\ref{QCS-Gauss-sparse}(middle) with a clear gap between Bernoulli
and Gaussian QCS performances when $K\leq 16$. CoBP$_\lambda$ does lead to clear improvements over CoBP.    

\subsection{Gaussian QCS of rank-1 matrices} 

We reconstruct here rank-1 matrices in $\bb
R^{32\times 32}$ (\ie $N=1024$ and $n=32$) from the Gaussian QCS model
\eqref{eq:psi-def} with $B \in \{1,2\}$. Both CoBP
and BPDN are solved with ${\|\!\cdot\!\|}_\sharp =
{\|\!\cdot\!\|}_*$. The intrinsic complexity of such rank-1 matrices is $63 < P:=64$.
For each value of the oversampling ratio $M/P \in [4,32]$, we generate 20 different
Gaussian sensing
matrices, dithering realizations and rank-1
matrices according to $\bs x_0 = \ve(\bs X_0)$ and $\bs X_0 = \bs v
\bs v^T/\|\bs v\|_2^2$ with $\bs v \sim \cl N^n(0,1)$. As for the first
experiment on $K$-sparse signals, CoBP reaches a faster reconstruction
error decay than BPDN. At $B=2$, an indicative
linear fitting over the last 4 values of $M/K$ provides estimated
decay exponents for CoBP and BPDN of $-0.85$ and $-0.33$,
respectively. 

\section{Conclusion}
\label{sec:conclusion}
In the context of QCS of signals with low-complexity (\eg sparse signals, low-rank
matrices), we show that the consistent reconstruction method CoBP
has an estimation error decaying as $M^{-1/4}$, \ie faster than the
one of BPDN. This is confirmed numerically on several settings
with even faster effective decaying rate at quantization resolution as
low as one bit per measurement. As observed initially in
1-bit~CS~\cite{ai2014one}, QCS performances for general sub-Gaussian sensing matrices are also impacted when the sensed signal is ``too
sparse''.  Finally, to the best of our knowledge, we provided the first
theoretical analysis of CoBP in the case of low-rank
matrix reconstruction from QCS observations.

\end{document}